\documentclass{article}
\usepackage{arxiv}
\usepackage{natbib}
\usepackage[T1]{fontenc}
\usepackage{amsmath,amssymb,amsthm}
\usepackage{graphicx}
\usepackage{xcolor}
\usepackage{booktabs}
\usepackage{multirow}
\usepackage{adjustbox}
\usepackage{url}
\usepackage[ruled,vlined,linesnumbered]{algorithm2e}

\SetCommentSty{algcommentfont}
\DontPrintSemicolon
\usepackage{hyperref}
\usepackage{enumitem}
\usepackage{caption}
\usepackage{placeins}
\usepackage{tikz}
\usetikzlibrary{decorations.pathreplacing, shapes.geometric, positioning, calc}

\theoremstyle{plain}
\newtheorem{proposition}{Proposition}[section]
\theoremstyle{definition}
\newtheorem{definition}{Definition}[section]
\theoremstyle{plain}

\newtheorem{remark}{Remark}

\makeatletter

\renewcommand{\sectionautorefname}{\S\@gobble}
\renewcommand{\subsectionautorefname}{\S\@gobble}
\renewcommand{\subsubsectionautorefname}{\S\@gobble}
\def\appendixautorefname{\S\@gobble}%

\makeatother

\newcommand{\eg}{\emph{e.g.}}

\newcommand{\topm}{\operatorname{top}_m}
\newcommand{\tool}{mBFV}   % m-BitFlip Verifier: top-m bound propagation + branch and bound
\newcommand{\toolnobab}{\tool{}\textsubscript{w/o BaB}}   % mBFV without branch and bound

\title{Towards Verifying Neural Networks Against Multi-Parameter Bit-Flip Perturbations}

\author{
  Hai Duong$^{1}$ \quad Thanh Le$^{2}$ \quad Ho Nguyen$^{1}$ \quad ThanhVu Nguyen$^{1}$ \\
  $^{1}$George Mason University, Fairfax, VA, USA \quad $^{2}$Unaffiliated, Yokosuka, Japan
}

\date{}

\renewcommand{\undertitle}{}
\renewcommand{\shorttitle}{Verifying Neural Networks Against Multi-Parameter Bit Flips}

\hypersetup{
  pdftitle={Towards Verifying Neural Networks Against Multi-Parameter Bit-Flip Perturbations},
  pdfauthor={Hai Duong, Thanh Le, Ho Nguyen, ThanhVu Nguyen},
}

\begin{document}
\maketitle

\begin{abstract}
Hardware faults can flip bits in the stored weights of a quantized neural network, potentially compromising its predictions. While such faults typically affect multiple parameters simultaneously, existing verifiers are limited to single-parameter perturbations due to the combinatorial explosion of possible flip locations in large networks. We present \tool{} ($m$-BitFlip Verifier), an efficient verification framework that proves robustness against simultaneous bit flips across multiple parameters without explicitly enumerating these combinations. \tool{} achieves this via a novel multi-parameter bound propagation technique that directly aggregates the $m$ worst-case contributions. To further tighten these bounds, \tool{} employs a branch-and-bound mechanism over perturbation locations, partitioning the potential flips to smaller groups of neurons. Evaluated on 625 instances, \tool{} successfully verifies 293, significantly outperforming a prior single-parameter verifier (38 verified instances) and an exact mixed-integer linear programming baseline (0 verified instances). Notably, while these baselines are restricted to single-parameter flips on small networks (up to 13k parameters), \tool{} scales to verify networks with up to 1.15M parameters against up to four simultaneous parameter flips.
\end{abstract}

\section{Introduction}\label{sec:intro}

Neural networks are widely deployed with quantized weights across diverse hardware platforms, ranging from edge microcontrollers~\citep{jacob2018quantization} to data center accelerators~\citep{frantar2023gptq}, and to wireless systems such as learned schedulers~\citep{le2025fggm}, semantic communication~\citep{le2026verifying}, traffic prediction~\citep{le2026formal}, and power control~\citep{le2026formal2}.
However, the physical memory storing these discrete values is vulnerable to hardware level bit flips, whether induced maliciously via DRAM Rowhammer~\citep{kim2014rowhammer} and fault injection~\citep{breier2018practical}, or naturally through radiation soft errors~\citep{schmedding2024aspis}.
Consequently, a few flipped bits in the weights can cause a network to misclassify~\citep{hong2019terminal,rakin2019bfa}.
Empirical testing is necessary but insufficient to show that such flips do not exist because the finite set of flip combinations grows exponentially with the number of flipped parameters.
In contrast, neural network verification (NNV) can prove that a property holds for every possible perturbation~\citep{katz2017reluplex,liu2021algorithms,singh2019deeppoly,zhang2018crown}.

BFAVerifier~\citep{zhang2025bfaverifier} is currently the state-of-the-art (SOTA) verifier for bit flips, yet it is strictly limited to a single perturbed parameter.
However, a single flipped parameter rarely changes the prediction of a quantized network~\citep{hong2019terminal}.
Known single-flip vulnerabilities need stronger assumptions such as control of training~\citep{dong2023onebit} or a planted input trigger~\citep{li2025backdoor}.
In practice, flips that change the prediction come in chains of 11 to 17 flips~\citep{rakin2019bfa} or 2 to 24 flips~\citep{yao2020deephammer}.
Therefore, a verification certificate for a single parameter can give false assurance.
In fact, in our evaluation (\autoref{sec:rq1}), only 74 of the 121 instances that \tool{} verifies with one flipped parameter remain verified with four.

Verifying multiple flipped parameters is inherently more challenging because there are $\binom{K}{m}$ ways to choose $m$ of $K$ parameters (\autoref{sec:bab-sound}).
For $m = 2$, a network with 5,000 parameters already has about 12.5 million pairs of parameters.
Running BFAVerifier~\citep{zhang2025bfaverifier} once per pair is infeasible even for such a small network.
A \emph{complete verifier} decides every instance given enough time, whereas an incomplete verifier certifies only some robust instances and returns unknown for the rest~\citep{tjeng2017evaluating}.
An MILP encoding~\citep{tjeng2017evaluating} is complete and covers all choices at once, but it does not scale beyond small networks (\autoref{sec:eval}).
To our knowledge, no existing verifier handles several flipped parameters on networks with millions of parameters.

In this paper, we introduce \tool{} ($m$ BitFlip Verifier), an efficient incomplete verifier that employs both bounding and branching to formally verify quantized networks against multiple parameter bit flip perturbations.
Our key observation is that bound propagation, which is the basis of verifiers such as CROWN~\citep{zhang2018crown} and DeepPoly~\citep{singh2019deeppoly}, can approximate the network output variation as a sum containing one distinct contribution per perturbed parameter.
When at most $m$ parameters may flip, the worst case output deviation is simply the aggregate of the $m$ largest individual contributions.
Therefore, our \emph{multiple-parameter bound propagation} algorithm computes this maximal sum directly without enumerating every possible parameter combination.
A classical result from combinatorial optimization~\citep{schrijver2003combinatorial} guarantees that this direct aggregation method incurs no loss of precision relative to the isolated parameter contributions.

To improve precision, \tool{} employs a branch and bound (BaB) strategy over the perturbed parameters.
Because each parameter contributes to exactly one neuron, altering $m$ parameters affects a maximum of $m$ neurons.
\tool{} leverages this structural property by isolating potential flips to small groups of neurons while keeping the remaining neurons to their original values, which allows multiple-parameter bound propagation to compute significantly tighter bounds.
Note that the generated subproblems cover the entire search space to ensure that no combination of flipped parameters is omitted.
The algorithm iteratively partitions any subproblem that bound propagation fails to verify until the entire space is proven safe.
More importantly, for a network containing $\nu$ neurons distributed across $L$ layers, this branching strategy generates a maximum of $1 + \binom{2\nu - L}{m}$ subproblems.
Since each subproblem requires only a single bound propagation, the overall verification complexity scales polynomially with the number of neurons (\autoref{prop:bab-count}).

We evaluate \tool{} againts $m \in \{1, 2, 4\}$ flipped parameters on fully-connected networks for ACAS Xu, MNIST, and CIFAR-10 with up to 1.15M parameters~\citep{brix2023first}.
Across 625 total instances, \tool{} with BaB proves 293 instances against 184 for \toolnobab{}, 38 for BFAVerifier, and none for the exact MILP.
Notably, while the SOTA baseline is restricted to single parameter perturbations on small networks containing at most 13k parameters, \tool{} scales effectively to verify architectures with up to 1.15M parameters and up to four simultaneous parameter flips.

Our contributions include:
(i) multiple parameter bound propagation, the core bounding mechanism of \tool{}, which evaluates $m$ flipped parameters simultaneously by aggregating their $m$ largest individual contributions (\autoref{sec:topm}),
(ii) a branch and bound strategy over perturbation locations that restricts potential flips to small groups of neurons, ensuring the total number of bounded subproblems scales polynomially with the number of neurons (\autoref{sec:bab}), and
(iii) an GPU implementation and an evaluation on architectures containing up to 1.15 million parameters (\autoref{sec:eval}).

\section{Background}\label{sec:problem}

\paragraph{Quantized neural networks.}
A quantized neural network (QNN) with $L$ layers and $K$ parameters is a function $N : \mathbb{R}^{d} \to \mathbb{R}^{s}$.
Each parameter $w$ of layer $i$ is stored as the signed $Q$-bit integer
\begin{equation}\label{eq:quant}
  q(w) = \mathrm{clip}\bigl(\mathrm{round}(w / \Delta_i),\ -(2^{Q-1}-1),\ 2^{Q-1}-1\bigr)
\end{equation}
where the step size $\Delta_i$ is the largest absolute value of a real parameter of layer $i$ divided by $2^{Q-1} - 1$, $\mathrm{round}$ rounds to the nearest integer, and the clip applies only at quantization time.
The integer is stored in two's complement $\vec v(q) = (v_Q, \dots, v_1) \in \{0,1\}^Q$ with sign bit $v_Q$ and $q = -2^{Q-1} v_Q + \sum_{t<Q} 2^{t-1} v_t$, and it is de-quantized to $\Delta_i\, q(w)$ at inference.
With the de-quantized weights $W^{(i)}$ and biases $b^{(i)}$, the network computes $z_0 = x$, $\hat z_i = W^{(i)} z_{i-1} + b^{(i)}$ (pre-activation), $z_i = \sigma(\hat z_i)$ (post-activation) with an activation $\sigma$ such as ReLU for $i < L$, and $N(x) = \hat z_L$, where layer $i$ has $\nu_i$ neurons.

\begin{definition}[Perturbation vector]\label{def:perturbation}
Given integers $m$ and $n \le Q$, an \emph{$(m,n)$-perturbation vector} $\rho = \{(\alpha_t, P_t) \mid t \le m\}$ is a set of at most $m$ pairs with pairwise distinct parameters $\alpha_t$ (weights or biases of $N$).
Each $P_t \subseteq \{1, \dots, Q\}$ is a possibly empty set of at most $n$ bit positions of $\vec v(q(\alpha_t))$.
The network $N^\rho$ is obtained from $N$ by flipping these bits, and $\mathcal{P}_{m,n}$ denotes the set of all $(m,n)$-perturbation vectors.
\end{definition}

\noindent Under at most $n$ bit flips, a parameter $\alpha$ of layer $i$ can take exactly the values
\begin{equation}\label{eq:reach}
  R_n(\alpha) = \bigl\{\, \Delta_i\, \mathrm{dec}\bigl(\vec v(q(\alpha)) \oplus e\bigr) \;:\;
  e \in \{0,1\}^Q,\ |e| \le n \,\bigr\}
\end{equation}
where $\oplus$ is bitwise XOR, $|e|$ counts the flipped bits, and $\mathrm{dec}$ decodes two's complement.
This set contains the \emph{nominal} value $\Delta_i\, q(\alpha)$ without flips ($e = 0$) and sign-bit flips, which can reach $-2^{Q-1}\Delta_i$ beyond the clip range of~\autoref{eq:quant}.
We enumerate the set exactly and use its interval hull, the \emph{change interval}
\begin{equation}\label{eq:hull}
  [\ell_\alpha, u_\alpha] = \bigl[\min R_n(\alpha) - \Delta_i\, q(\alpha),\ \max R_n(\alpha) - \Delta_i\, q(\alpha)\bigr],
  \qquad \ell_\alpha \le 0 \le u_\alpha
\end{equation}

\paragraph{DNN Verification.}
NNV techniques prove that every input satisfying a property $\phi_{in}$ yields an output satisfying a property $\phi_{out}$.
Standard NNV problems take $\phi_{in}$ as a box $\mathcal{B} = \{x : \|x - x_0\|_\infty \le \epsilon\}$ and $\phi_{out}$ as $N(x)_g > N(x)_k$ for every $k \ne g$, where $g$ is the label of $x_0$.
The property thus holds if and only if the smallest margin $f(x) = \min_{k \ne g} \bigl(N(x)_g - N(x)_k\bigr)$ stays positive on the box
\begin{equation}\label{eq:nnv}
  \min_{x \in \mathcal{B}}\ f(x) \;>\; 0
\end{equation}
Under bit flips, we require~\autoref{eq:nnv} for every perturbed network $N^\rho$ with $\rho \in \mathcal{P}_{m,n}$:
\begin{equation}\label{eq:goal}
  f^\star \;:=\; \min_{\rho \in \mathcal{P}_{m,n}}\ \min_{x \in \mathcal{B}}\ f^\rho(x) \;>\; 0
\end{equation}
A verifier is \emph{sound} if it never verifies a problem with $f^\star \le 0$.
A \emph{complete} verifier also decides every problem, \eg, by solving~\eqref{eq:goal} exactly as an MILP (\autoref{app:milp}).
An \emph{incomplete} verifier such as \tool{} computes a lower bound $lb \le f^\star$. If $lb > 0$, it returns \emph{verified}, since no perturbation $\rho \in \mathcal{P}_{m,n}$ and input $x \in \mathcal{B}$ violate the property. Otherwise it returns \emph{unknown}.

\section{Multi-Parameter Bound Propagation}\label{sec:topm}

\tool{} extends CROWN~\citep{zhang2018crown}, which bounds each neuron as $\underline A x + \underline b \le \hat z_j \le \overline A x + \overline b$, so that these bounds also hold when up to $m$ parameters flip.
We call this extension \emph{multi-bitflip bound propagation} (\autoref{alg:topm}).
At each layer $\ell$ (\autoref{ln:layer}), \tool{} bounds every pre-activation neuron (\autoref{ln:rows}) by \emph{back-substitution} to the input (\autoref{ln:back}).
During the back-substitution, it records the smallest and the largest term (\autoref{eq:g}) of every parameter in $\underline G$ and $\overline G$ (\autoref{ln:terms}-\autoref{ln:termsup}, \autoref{sec:terms}) and relaxes the activation (\autoref{ln:relax}).
Next, it \emph{concretizes} both bounds over $\phi_{in}$ (\autoref{sec:choice}) and adds the sum of the $m$ smallest terms to the lower bound (\autoref{ln:conc}) and the sum of the $m$ largest terms to the upper bound (\autoref{ln:concup}).
Activation $\textsc{Act}$ then maps the resulting pre-activation ranges for the next layer (\autoref{ln:act}).
Finally, \tool{} returns the lower bound of $\phi_{out}$ (\autoref{ln:ret}).

\begin{algorithm}[t]
\caption{Bound propagation under up to $m$ flipped parameters}\label{alg:topm}
\KwIn{network $N$ (weights $W$, biases $b$), property $\phi_{in} \Rightarrow \phi_{out}$, parameter intervals $[\ell_p, u_p]$, numbers of perturbed parameters $m$}
\KwOut{$lb \le f^\star$}
$[\underline z_0, \overline z_0] \gets \phi_{in}$\;\label{ln:init}
\For{$\ell = 1, \dots, L$\label{ln:layer}}{
  \ForEach{neuron $t$ of layer $\ell$\label{ln:rows}}{
    $\underline G, \overline G \gets \emptyset$\;\label{ln:start}
    \For{$i = \ell, \dots, 1$\label{ln:back}}{
      \ForEach{parameter $p = (j, k)$ of layer $i$}{
        $\underline G \gets \underline G \cup \bigl\{\min\{\underline A_j\, D\, z : D \in \{\ell_p, u_p\},\ z \in \{\underline z_{i-1,k}, \overline z_{i-1,k}\}\}\bigr\}$ \tcp*{\autoref{eq:g}}\label{ln:terms}
        $\overline G \gets \overline G \cup \bigl\{\max\{\overline A_j\, D\, z : D \in \{\ell_p, u_p\},\ z \in \{\underline z_{i-1,k}, \overline z_{i-1,k}\}\}\bigr\}$\;\label{ln:termsup}
      }
      $(\underline A, \underline b, \overline A, \overline b) \gets \textsc{Relax}\bigl(\underline A, \underline b, \overline A, \overline b, [\underline{\hat z}_{i-1}, \overline{\hat z}_{i-1}]\bigr)$\;\label{ln:relax}
    }
    $\underline{\hat z}_{\ell,t} \gets \min_{x \in \phi_{in}} (\underline A x + \underline b) + \topm(\underline G)$ \tcp*{\autoref{prop:topm}}\label{ln:conc}
    $\overline{\hat z}_{\ell,t} \gets \max_{x \in \phi_{in}} (\overline A x + \overline b) - \topm(-\overline G)$\;\label{ln:concup}
  }
  $[\underline z_\ell, \overline z_\ell] \gets \textsc{Act}\bigl([\underline{\hat z}_\ell, \overline{\hat z}_\ell]\bigr)$\;\label{ln:act}
}
\Return{$\min_t \underline{\hat z}_{L,t}$}\;\label{ln:ret}
\end{algorithm}

\subsection{Per-Parameter Decomposition}\label{sec:terms}
Back-substitution bounds an output $\hat z_{\ell,t}$ in terms of each earlier layer $i$:
\begin{equation}\label{eq:layerbound}
  \underline A\, \hat z_i + \underline b \;\le\; \hat z_{\ell,t} \;\le\; \overline A\, \hat z_i + \overline b
\end{equation}
The process starts at $i = \ell$ and ends at the input $x$, where the bounds become linear functions of $x$, \eg, $\underline A x + \underline b  \;\le\; \hat z_{\ell,t} \;\le\; \overline A x + \overline b$.
Each step from layer $i$ to layer $i-1$ replaces $\hat z_i$ in~\autoref{eq:layerbound} by its definition, which under bit flips is
\begin{equation}\label{eq:neuron}
  \hat z_i = (W^{(i)} + D^{(i)})\, z_{i-1} + b^{(i)} + d^{(i)}
\end{equation}
where each change $D^{(i)}_{jk}$ or $d^{(i)}_j$ is $0$ if its parameter $p$ is not perturbed and otherwise lies in its change interval $[\ell_p, u_p]$~\autoref{eq:hull}.
Substituting~\autoref{eq:neuron} into~\autoref{eq:layerbound} gives
\begin{equation}\label{eq:subst}
  \begin{aligned}
  \underline A\, \hat z_i + \underline b &= \underbrace{\underline A \big(W^{(i)} z_{i-1} + b^{(i)}\big) + \underline b}_{\text{nominal}}
  + \sum_{j,k} \underline A_j D^{(i)}_{jk} z_{i-1,k} + \sum_j \underline A_j d^{(i)}_j \\
  \overline A\, \hat z_i + \overline b &= \overline A \big(W^{(i)} z_{i-1} + b^{(i)}\big) + \overline b
  + \sum_{j,k} \overline A_j D^{(i)}_{jk} z_{i-1,k} + \sum_j \overline A_j d^{(i)}_j
  \end{aligned}
\end{equation}
The nominal part is back-substituted as in CROWN~\citep{zhang2018crown}, while the remaining terms carry the flips.
Since the flip of each weight $p = W^{(i)}_{jk}$ is unknown, \tool{} replaces its terms $\underline A_j D^{(i)}_{jk} z_{i-1,k}$ and $\overline A_j D^{(i)}_{jk} z_{i-1,k}$ by their worst cases
\begin{equation}\label{eq:g}
  \begin{aligned}
  \underline g_p &= \min_{D^{(i)}_{jk} \in [\ell_p, u_p],\; z_{i-1,k} \in [\underline z_{i-1,k}, \overline z_{i-1,k}]} \underline A_j\, D^{(i)}_{jk}\, z_{i-1,k} \\
  \overline g_p &= \max_{D^{(i)}_{jk} \in [\ell_p, u_p],\; z_{i-1,k} \in [\underline z_{i-1,k}, \overline z_{i-1,k}]} \overline A_j\, D^{(i)}_{jk}\, z_{i-1,k}
  \end{aligned}
\end{equation}
Both extremes lie at corners of $[\ell_p, u_p] \times [\underline z_{i-1,k}, \overline z_{i-1,k}]$ because the term is bilinear (\autoref{ln:terms}).
Likewise, the terms $\underline A_j d^{(i)}_j$ and $\overline A_j d^{(i)}_j$ of each bias $p = b^{(i)}_j$ have the worst cases
\begin{equation}\label{eq:gbias}
  \underline g_p = \min_{d^{(i)}_j \in [\ell_p, u_p]} \underline A_j\, d^{(i)}_j
  \qquad
  \overline g_p = \max_{d^{(i)}_j \in [\ell_p, u_p]} \overline A_j\, d^{(i)}_j
\end{equation}
Let the vector $s \in \{0,1\}^K$ indicate the perturbed parameters, with $s_p = 1$ if parameter $p$ is perturbed and $s_p = 0$ otherwise.
For every input $x \in \mathcal{B}$ and every choice of perturbed parameters, \tool{} bounds $\hat z_{\ell,t}$ by
\begin{equation}\label{eq:decomp}
  \underline A x + \underline b + \sum_p s_p\, \underline g_p \;\le\; \hat z_{\ell,t} \;\le\; \overline A x + \overline b + \sum_p s_p\, \overline g_p
\end{equation}

\subsection{Multi-Parameter Concretization}\label{sec:choice}
Concretization turns~\autoref{eq:decomp} into scalars by optimizing over every input $x \in \mathcal{B}$ and every choice of at most $m$ perturbed parameters $\mathcal{S}_m = \Big\{ s \in \{0,1\}^K : \sum_p s_p \le m \Big\}$.
The input part and the perturbation part of~\autoref{eq:decomp} share no variables, so the bounds split into:
\begin{equation}\label{eq:conc}
  \begin{aligned}
  \underline{\hat z}_{\ell,t} &= \underbrace{\min_{x \in \mathcal{B}} \big(\underline A x + \underline b\big)}_{\text{input part}}
  + \underbrace{\min_{s \in \mathcal{S}_m} \sum_p s_p\, \underline g_p}_{\text{perturbation part}} \\
  \overline{\hat z}_{\ell,t} &= \max_{x \in \mathcal{B}} \big(\overline A x + \overline b\big)
  + \max_{s \in \mathcal{S}_m} \sum_p s_p\, \overline g_p
  \end{aligned}
\end{equation}
The input part is concretized as in CROWN~\citep{zhang2018crown}.
The perturbation part is an integer program over $\binom{K}{m}$ choices of parameters, which \tool{} relaxes to the linear program over the convex hull of $\mathcal{S}_m$:
\begin{equation}\label{eq:relax}
  \mathcal{C}_m = \Big\{ s \in [0,1]^K : \sum_p s_p \le m \Big\}
\end{equation}
\autoref{prop:topm} shows that the relaxation is exact, so linear and integer programs have the same optimum:
\begin{equation}\label{eq:topm}
  \begin{aligned}
  \min_{s \in \mathcal{C}_m} \sum_p s_p\, \underline g_p &= \min_{s \in \mathcal{S}_m} \sum_p s_p\, \underline g_p = \topm(\underline g) \\
  \max_{s \in \mathcal{C}_m} \sum_p s_p\, \overline g_p &= \max_{s \in \mathcal{S}_m} \sum_p s_p\, \overline g_p = -\topm(-\overline g)
  \end{aligned}
\end{equation}
where $\topm(g)$ is the sum of the $m$ smallest entries of $g$, and the upper bound follows by applying \autoref{prop:topm} to $-\overline g$.

\begin{proposition}\label{prop:topm}
Assume $g_p \le 0$ for every $p$. Let $T$ be a set of $\min(m, K)$ parameters whose terms are no
larger than any term outside $T$ (the $m$ most negative terms). Then
\begin{enumerate}[nosep,label=(\alph*)]
  \item $\sum_{p \in T} g_p \le \sum_p s_p g_p$ for every $s \in \mathcal{C}_m$;
  \item the 0/1 vector $\mathbf{1}_T$ of $T$ lies in $\mathcal{S}_m$ and attains $\sum_{p \in T} g_p$.
\end{enumerate}
Hence
$\displaystyle \min_{s \in \mathcal{C}_m} \sum_p s_p g_p
 = \min_{s \in \mathcal{S}_m} \sum_p s_p g_p = \sum_{p \in T} g_p = \topm(g)$.
\end{proposition}
\begin{proof}[Proof sketch]
Pick a threshold $\theta \le 0$ with $g_p \le \theta$ inside $T$ and $g_p \ge \theta$ outside. Then
$(s_p - \mathbf{1}_T(p))(g_p - \theta) \ge 0$ elementwise, and summing gives
$\sum_p s_p g_p - \sum_{p \in T} g_p \ge \theta (\sum_p s_p - |T|) \ge 0$. The vector $\mathbf{1}_T$
attains the bound. \autoref{app:proof} gives the full proof, checked in Lean.
\end{proof}
For example, with $g = (-5, -3, -1, 0)$ and $m=2$, the optimum is $s = (1,1,0,0)$ with value $-8$, while $s = (1, 0.5, 0.5, 0)$ reaches only $-7$.
The perturbation part thus has a closed form, whereas word substitutions need dynamic programming~\citep{xu2020automatic}.

\subsection{Soundness of \tool{}}\label{sec:budget}

We show that every bound of \autoref{alg:topm} holds for every input $x \in \mathcal{B}$ and every flip in at most $m$ parameters.
Assume this holds for the neuron ranges of the earlier layers, which \autoref{alg:topm} computes first, so CROWN's relaxation of the activation is valid.
Each flipped parameter then contributes a term between $\underline g_p$ and $\overline g_p$ to~\autoref{eq:subst}, since its change lies in $[\ell_p, u_p]$ and its input lies in its range.
Every other parameter contributes $0$, so~\autoref{eq:decomp} holds with the vector $s \in \mathcal{S}_m$ of the flipped parameters.
Because~\autoref{eq:conc} optimizes over all inputs and all $s \in \mathcal{S}_m$, the true value of $\hat z_{\ell,t}$ lies between $\underline{\hat z}_{\ell,t}$ and $\overline{\hat z}_{\ell,t}$.
By induction over the layers, \tool{} is therefore sound.

\section{Branch and Bound for Bit Flip Verification}\label{sec:bab}

Branch and bound tightens the bounds of multi-bitflip bound propagation (\autoref{sec:topm}) by restricting where the flips may happen.
Each \emph{subproblem} allows flips only in a few chosen neurons, so every other neuron keeps its original range (\autoref{sec:branch}).
Fewer changing neurons give tighter ranges in the following layers and hence a tighter bound on the margin.
The network is verified when every subproblem is verified, and each subproblem that is not verified is split into smaller ones.

\autoref{alg:bab} describes the search, where $\textsc{Bound}$ is the lower bound on the margin that multi-bitflip bound propagation computes when flips are allowed only in the given neurons.
It first allows flips in all neurons (\autoref{ln:bab-root}) and otherwise bounds the subproblems in a queue (\autoref{ln:bab-init} to \autoref{ln:bab-bound}).
It splits every subproblem whose bound is not positive (\autoref{ln:bab-branch}) and returns verified once all subproblems are verified (\autoref{ln:bab-ret}).

\begin{algorithm}[t]
\caption{Branch and bound over perturbation locations}\label{alg:bab}
\KwIn{network $N$, property $\phi_{in} \Rightarrow \phi_{out}$, number $m$ of flipped parameters, time limit}
\KwOut{verified, unknown, or timeout}
\If(\tcp*[f]{\autoref{sec:topm}}){$\textsc{Bound}(\text{all neurons}) > 0$\label{ln:bab-root}}{\Return{verified}\;}
$Q \gets \{ S \subseteq \{\text{layers}\} : |S| = m \}$, \quad $\mathcal{G} \gets Q$, \quad $P \gets \emptyset$ \tcp*{\autoref{sec:branch}}\label{ln:bab-init}
\While{$Q \neq \emptyset$}{
  pop $S \in Q$ with $\bigcup S \not\subseteq \bigcup S'$ for all $S' \in P$ \tcp*{\autoref{sec:bab-sound}}
  \eIf(\tcp*[f]{\autoref{sec:topm}}){$\textsc{Bound}(S) > 0$\label{ln:bab-bound}}{
    $P \gets P \cup \{S\}$\;
  }{
    \lIf{$|G| = 1$ for all $G \in S$}{\Return{unknown}}\label{ln:bab-leaf}
    $Q \gets Q \cup (\textsc{Branch}(S) \setminus \mathcal{G})$, \quad $\mathcal{G} \gets \mathcal{G} \cup \textsc{Branch}(S)$ \tcp*{\autoref{eq:branch}}\label{ln:bab-branch}
  }
  \lIf{time limit reached}{\Return{timeout}}
}
\Return{verified} \label{ln:bab-ret}
\end{algorithm}

\subsection{Groups and Subproblems}\label{sec:branch}
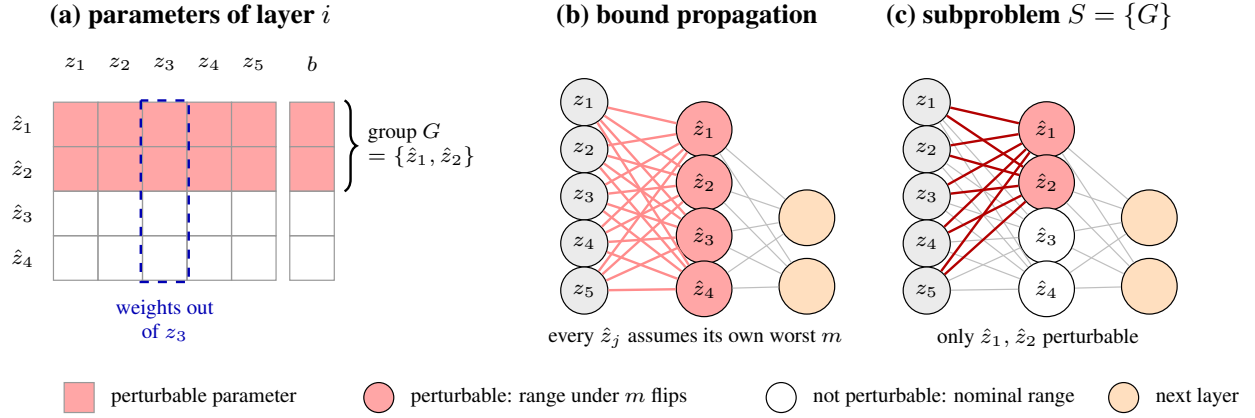
\begin{figure}[t]
\centering
\resizebox{\linewidth}{!}{%
\begin{tikzpicture}[
    cell/.style={draw=black!40, minimum size=5.2mm, inner sep=0pt},
    att/.style={cell, fill=red!35},
    neu/.style={circle, draw, minimum size=6.5mm, inner sep=0pt, font=\scriptsize},
    wide/.style={neu, fill=red!35},
    exact/.style={neu, fill=white},
    down/.style={neu, fill=orange!25},
    inp/.style={circle, draw, fill=black!8, minimum size=5.5mm, inner sep=0pt, font=\scriptsize},
    edge/.style={black!25},
    aedge/.style={red!70!black, thick}]
\begin{scope}
  \node[font=\small\bfseries] at (1.6, 1.0) {(a) parameters of layer $i$};
  \foreach \j in {1,...,5} \node[font=\scriptsize] at (0.52*\j-0.26, 0.45) {$z_{\j}$};
  \node[font=\scriptsize] at (3.02, 0.45) {$b$};
  \foreach \i in {1,...,4} {
    \node[font=\scriptsize, anchor=east] at (-0.1, -0.52*\i+0.26) {$\hat z_{\i}$};
    \foreach \j in {1,...,5} {
      \ifnum\i<3 \node[att] at (0.52*\j-0.26, -0.52*\i+0.26) {};
      \else \node[cell] at (0.52*\j-0.26, -0.52*\i+0.26) {}; \fi }
    \ifnum\i<3 \node[att] at (3.02, -0.52*\i+0.26) {};
    \else \node[cell] at (3.02, -0.52*\i+0.26) {}; \fi }
  \draw[decorate, decoration={brace, amplitude=4pt}, thick]
        (3.4, 0.02) -- (3.4, -1.04) node[midway, right=4pt, align=left, font=\scriptsize]
        {group $G$\\$= \{\hat z_1, \hat z_2\}$};
  \draw[blue!70!black, dashed, thick] (1.04-0.02, 0.02) rectangle (1.56+0.02, -2.10);
  \node[font=\scriptsize, blue!70!black, align=center] at (1.3, -2.55) {weights out\\of $z_3$};
\end{scope}
\begin{scope}[xshift=6.2cm]
  \node[font=\small\bfseries] at (1.2, 1.0) {(b) bound propagation};
  \foreach \j in {1,...,5} \node[inp] (x\j) at (0, 0.55-0.55*\j) {$z_{\j}$};
  \foreach \i in {1,...,4} \node[wide] (z\i) at (1.4, 0.3-0.62*\i) {$\hat z_{\i}$};
  \foreach \i in {1,...,4} \foreach \j in {1,...,5} \draw[aedge, red!45] (x\j) -- (z\i);
  \foreach \k in {1,2} \node[down] (y\k) at (2.6, -0.55-0.8*\k) {};
  \foreach \i in {1,...,4} \foreach \k in {1,2} \draw[edge] (z\i) -- (y\k);
  \node[font=\scriptsize, align=center] at (1.3, -2.75) {every $\hat z_j$ assumes its own worst $m$};
\end{scope}
\begin{scope}[xshift=10.2cm]
  \node[font=\small\bfseries] at (1.2, 1.0) {(c) subproblem $S = \{G\}$};
  \foreach \j in {1,...,5} \node[inp] (x\j) at (0, 0.55-0.55*\j) {$z_{\j}$};
  \node[wide]  (z1) at (1.4, 0.3-0.62*1) {$\hat z_1$};
  \node[wide]  (z2) at (1.4, 0.3-0.62*2) {$\hat z_2$};
  \node[exact] (z3) at (1.4, 0.3-0.62*3) {$\hat z_3$};
  \node[exact] (z4) at (1.4, 0.3-0.62*4) {$\hat z_4$};
  \foreach \i in {3,4} \foreach \j in {1,...,5} \draw[edge] (x\j) -- (z\i);
  \foreach \i in {1,2} \foreach \j in {1,...,5} \draw[aedge] (x\j) -- (z\i);
  \foreach \k in {1,2} \node[down] (y\k) at (2.6, -0.55-0.8*\k) {};
  \foreach \i in {1,...,4} \foreach \k in {1,2} \draw[edge] (z\i) -- (y\k);
  \node[font=\scriptsize, align=center] at (1.3, -2.75) {only $\hat z_1, \hat z_2$ perturbable};
\end{scope}
\begin{scope}[yshift=-3.45cm, xshift=0.3cm]
  \node[att, minimum size=3.5mm] at (0,0) {}; \node[font=\scriptsize, anchor=west] at (0.25,0) {perturbable parameter};
  \node[wide, minimum size=3.5mm] at (3.5,0) {}; \node[font=\scriptsize, anchor=west] at (3.75,0) {perturbable: range under $m$ flips};
  \node[exact, minimum size=3.5mm] at (8.2,0) {}; \node[font=\scriptsize, anchor=west] at (8.45,0) {not perturbable: nominal range};
  \node[down, minimum size=3.5mm] at (12.2,0) {}; \node[font=\scriptsize, anchor=west] at (12.45,0) {next layer};
\end{scope}
\end{tikzpicture}}
\caption{Allowing flips only in a group of neurons of layer $i$.}
\label{fig:neurons}
\end{figure}

The parameters of a neuron $\hat z_j$ are its incoming weights and bias, and they feed only $\hat z_j$ (\autoref{fig:neurons}a).
Since the linear relaxation allow flips to occur anywhere, multi-bitflip bound propagation therefore widens the range of every neuron (\autoref{fig:neurons}b).
Consider first $m$ flips in a single layer.
A \emph{group} is a set of neurons of one layer, and allowing flips only in a group widens only its neurons (\autoref{fig:neurons}c).
All other neurons keep their original ranges, which tightens the ranges in the following layers.
Flips spread over several layers can be solved by solving multiple a single layer multi-bitflip.

A \emph{subproblem} $S$ is a set of at most $m$ disjoint groups, and a perturbation \emph{lies in} $S$ if it flips only parameters of neurons in these groups.
$\textsc{Bound}(S)$ allows flips only in these neurons, so it bounds the margin under every perturbation that lies in $S$.
The search starts with whole layers as groups (\autoref{ln:bab-init}), since $m$ flipped parameters lie in at most $m$ layers.

When $S$ is not verified, $\textsc{Branch}(S)$ splits its largest group $G$ at the middle neuron into two parts $G_1$ and $G_2$ and returns the subproblems
\begin{equation}\label{eq:branch}
  \textsc{Branch}(S) = \{ C \subseteq T : |C| = m \}, \qquad T = (S \setminus \{G\}) \cup \{G_1, G_2\}
\end{equation}
Each of the $m$ flipped parameters lies in exactly one group of $T$, so every perturbation that lies in $S$ also lies in one of these subproblems (\autoref{prop:cover-step}).
If $T$ has at most $m$ groups, the only such subproblem allows the same neurons as $S$, so $\textsc{Branch}$ splits again.
Unlike branching on activation phases~\citep{wang2021beta} or on the input box, this search branches on the remaining multi-bitflip relaxed neurons. % where the flips happen.

\subsection{Soundness and Complexity}\label{sec:bab-sound}
\paragraph{Soundness.} At every step of \autoref{alg:bab}, each perturbation in $\mathcal{P}_{m,n}$ lies in a subproblem of the queue $Q$ or of the verified set $P$, since $\textsc{Branch}$ keeps every perturbation (\autoref{prop:cover}).
A subproblem $S \in P$ has $\textsc{Bound}(S) > 0$, so the margin stays positive under every perturbation that lies in $S$.
Let $S'$ be a subproblem in $Q$ whose neurons are all neurons of some $S \in P$.
Every perturbation that lies in $S'$ then also lies in $S$ and is already verified, so the search skips $S'$ without bounding it.
When $Q$ is empty, every perturbation therefore lies in a subproblem of $P$, so the algorithm is sound (\autoref{prop:bab-sound}).

\paragraph{Complexity.} The cost of \autoref{alg:bab} is the number of calls of $\textsc{Bound}$, each of which runs multi-bitflip bound propagation once.
\autoref{prop:bab-count} bounds this number in terms of the number of neurons $\nu$ (\autoref{sec:budget}).

\begin{proposition}[Number of bound propagations]\label{prop:bab-count}
For every $m$, batch size, time limit, and outcome of \textnormal{\textsc{Bound}}, and
whatever the verdict, \textnormal{\tool{}} makes at most
\begin{equation}\label{eq:bab-count}
  1 + \binom{2\nu - L}{m}
\end{equation}
bound propagations. Assume $\nu > m$, no time limit, and $\textnormal{\textsc{Bound}}(S) > 0$ exactly
when $S$ may perturb the parameters of at most $m$ neurons. Then \textnormal{\tool{}} returns verified
after at least $1 + \binom{\nu}{m}$ bound propagations, and for $m = 1$ and $L \ge 2$ after exactly
$1 + 2\nu - L$, so~\autoref{eq:bab-count} is attained.
\end{proposition}
\noindent The count is polynomial in $\nu$, whereas there are $\binom{K}{m}$ choices of $m$
parameters. For MNIST FC $2{\times}256$ ($\nu = 522$, $L = 3$) and $m = 2$, the bound is
$541{,}321$ bound propagations, against $3.6 \cdot 10^{10}$ pairs of parameters.
\autoref{app:bab-count} gives the proof, whose counting is checked in Lean.

\section{Evaluation}\label{sec:eval}

\subsection{Benchmarks and Setup}

\paragraph{Benchmarks and Fine-tuning.}
We evaluate \tool{} on seven architectures~\citep{brix2023first,kaulen20256th,zhang2025bfaverifier}: three ACAS Xu networks ($6{\times}50$), three fully connected MNIST networks ($2{\times}256$, $4{\times}256$, $7{\times}250$), and a CIFAR-10 network ($7{\times}250$, 1.15M parameters).

Following BFAVerifier~\citep{zhang2025bfaverifier}, all networks are quantized to $Q = 8$ bits with $n = 1$ flipped bit per perturbed parameter. Because the pretrained networks are highly vulnerable to single-bit flips (failing all pilot tests at $m = 1$), we fine-tune them~\citep{xu2024training,duong2024harnessing} using a hinge loss on the margin. To bound the perturbation intervals, we clip all weights and biases to $[-c, c]$ after each step. The clip bound $c$ is fixed prior to verification as the smallest value that maintains accuracy within 3\% of the pretrained network (\autoref{app:setup}).

\paragraph{Verification Instances.}

An verification instance pairs one input with $m \in \{1, 2, 4\}$ flipped parameters, evaluated under a 60-second timeout. The input perturbation radius is $0$ or $2/255$ ($0.01$ for ACAS Xu). To ensure fair evaluations, we follow BFAVerifier by selecting the first 15 correctly classified test inputs that bound propagation successfully verifies at $m = 0$. For ACAS Xu, inputs are sampled from the normalized domain using the network's advisory as the label.
No CIFAR-10 image passes this filter at radius $2/255$, so \autoref{tab:main} has 13 network--radius pairs (585 instances), and two pretrained networks add 40 (\autoref{app:setup}).

\paragraph{Baselines} We compare \tool{} against BFAVerifier and an exact MILP extended to $m$ flips. The MILP is the only complete method for $m \ge 2$, building a model with one set of binary variables per parameter; however, if model construction times out on the first instance of a run, all remaining instances for that run are reported as unproved (\autoref{app:setup}).

\subsection{\tool{}'s Performance Compared to Baselines}\label{sec:rq1}

\begin{table}[t]
\centering
\caption{Verified instances out of 15 inputs per network and radius within 60\,s.}
\label{tab:main}
\small
\setlength{\tabcolsep}{2pt}
\begin{adjustbox}{max width=\linewidth}
\begin{tabular}{lc cccc cccc cccc}
\toprule
 & & \multicolumn{4}{c}{$m = 1$} & \multicolumn{4}{c}{$m = 2$} & \multicolumn{4}{c}{$m = 4$} \\
\cmidrule(lr){3-6} \cmidrule(lr){7-10} \cmidrule(lr){11-14}
network & rad. & \toolnobab{} & \tool{} & BFAV & MILP & \toolnobab{} & \tool{} & BFAV & MILP & \toolnobab{} & \tool{} & BFAV & MILP \\
\midrule
\multirow{2}{*}{ACAS Xu 1\_1} & 0 & 0 & 4 & \textbf{5} & 0 & 0 & 0 & -- & -- & 0 & 0 & -- & -- \\
 & 0.01 & 0 & 4 & \textbf{7} & 0 & 0 & 0 & -- & -- & 0 & 0 & -- & -- \\
\midrule
\multirow{2}{*}{ACAS Xu 2\_1} & 0 & 0 & 3 & \textbf{8} & 0 & 0 & 0 & -- & -- & 0 & 0 & -- & -- \\
 & 0.01 & 0 & 0 & \textbf{8} & 0 & 0 & 0 & -- & -- & 0 & 0 & -- & -- \\
\midrule
\multirow{2}{*}{ACAS Xu 3\_3} & 0 & 0 & 4 & \textbf{5} & 0 & 0 & 0 & -- & -- & 0 & 0 & -- & -- \\
 & 0.01 & 0 & 2 & \textbf{5} & 0 & 0 & 0 & -- & -- & 0 & 0 & -- & -- \\
\midrule
\multirow{2}{*}{MNIST FC $2{\times}256$} & 0 & \textbf{15} & \textbf{15} & 0 & 0 & 14 & \textbf{15} & -- & -- & 13 & \textbf{14} & -- & -- \\
 & 2/255 & \textbf{15} & \textbf{15} & 0 & 0 & \textbf{14} & \textbf{14} & -- & -- & 5 & \textbf{14} & -- & -- \\
\midrule
\multirow{2}{*}{MNIST FC $4{\times}256$} & 0 & \textbf{15} & \textbf{15} & 0 & 0 & \textbf{15} & \textbf{15} & -- & -- & 1 & \textbf{15} & -- & -- \\
 & 2/255 & \textbf{15} & \textbf{15} & 0 & 0 & 5 & \textbf{15} & -- & -- & 0 & \textbf{14} & -- & -- \\
\midrule
\multirow{2}{*}{MNIST FC $7{\times}250$} & 0 & \textbf{15} & \textbf{15} & 0 & 0 & \textbf{15} & \textbf{15} & -- & -- & 2 & \textbf{13} & -- & -- \\
 & 2/255 & \textbf{15} & \textbf{15} & 0 & 0 & 5 & \textbf{15} & -- & -- & 0 & \textbf{4} & -- & -- \\
\midrule
CIFAR FC $7{\times}250$ & 0 & 5 & \textbf{14} & 0 & 0 & 0 & \textbf{9} & -- & -- & 0 & 0 & -- & -- \\
\midrule
\multicolumn{2}{l}{\textbf{Total}} & 95 & \textbf{121} & 38 & 0 & 68 & \textbf{98} & -- & -- & 21 & \textbf{74} & -- & -- \\
\bottomrule
\end{tabular}
\end{adjustbox}
\end{table}

\begin{figure}[t]
\centering
\includegraphics[width=\linewidth]{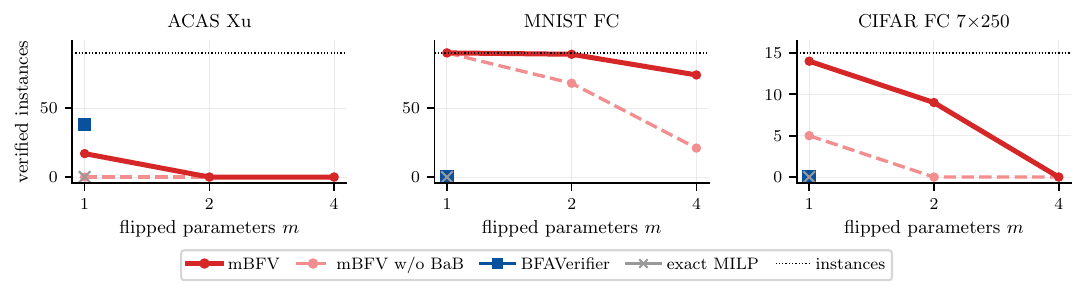}
\caption{Verified instances against the number $m$ of flipped parameters per network family.}
\label{fig:budget}
\end{figure}

\autoref{tab:main} compares \tool{} against an ablation without branch and bound (\toolnobab{}), BFAVerifier (BFAV), and an exact MILP baseline.
At $m = 1$, \tool{} verifies 121 total instances, compared to 38 for BFAV and 0 for MILP.
\tool{} performs particularly well on the larger MNIST and CIFAR networks (269k–1.15M parameters), verifying 104 of 105 instances while both baselines time out.

However, BFAV verifies more instances than \tool{} on the smaller ACAS Xu networks (13k parameters).
Because these networks have relatively few parameters, BFAV's single parameter enumeration remains computationally feasible.
Additionally, their large clip bound ($c = 0.5$) means single bit flips can substantially change outputs.
BFAV handles this shift precisely with exact MILP, whereas \tool{}'s relaxations become too loose in this specific setting.
\autoref{fig:budget} splits these results by network family and shows that branch and bound gains the most on the fully connected MNIST and CIFAR networks.

As illustrated in \autoref{fig:eval}a, BFAV and MILP do not scale beyond a single flip, \tool{} is the only method that is able to verify instances at $m=2$ (98 verified) and $m=4$ (74 verified).
Without BaB, \toolnobab{}'s performance decreases from 95 verified instances at $m=1$ to 21 at $m=4$.
In contrast, \tool{} partitions the search space to maintain tighter bounds, allowing it to verify more instances, \eg, on MNIST FC $4{\times}256$ ($m=4$, radius 0), \tool{} verifies all 15 instances while \toolnobab{} verifies only 1.
The exact MILP cannot verify any instances with $m \ge 2$.

\subsection{Effectiveness of Branch and Bound}\label{sec:rq2}

\begin{figure}[t]
  \centering
  \includegraphics[width=\linewidth]{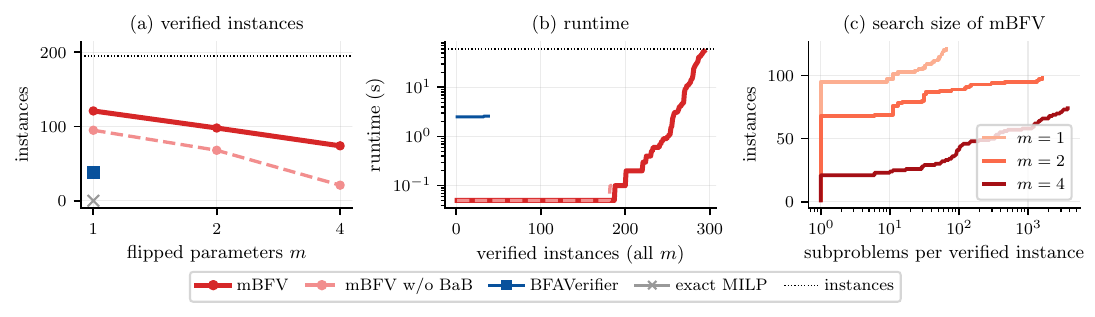}
  \caption{Verified instances, runtime, and search size of \tool{} over all networks.}
  \label{fig:eval}
  \end{figure}

\autoref{fig:eval}b presents verification runtimes across all numbers of flipped parameters.
\toolnobab{} performs only a single bound propagation, which makes it fast but cannot verify many instances.
\tool{} matches \toolnobab{} on easy instances, verifying about 190 instances in under 0.1 seconds. For more complex problems, \tool{} performs branching and computes additional bound propagations, which requires more runtime, scaling progressively up to the 60 second timeout.
BFAVerifier requires higher base latency, requiring over 2.5 seconds even for its easiest instances, and verifies only 38 instances in total, while the exact MILP does not verify any within the time limit.

\autoref{fig:eval}c illustrates the search space behavior by plotting the evaluated subproblems per verified instance.
The vertical lines at $10^0$ represents instances solved entirely at the root node by the continuous relaxation, matching the capabilities of \toolnobab{}.
At $m = 1$, this relaxation is tight enough to verify nearly 100 instances without branching.
As the number of flipped parameters increases, multiple parameter flips widen the downstream bounds, reducing the number of instances verified at the root to 68 at $m = 2$ and 21 at $m = 4$.
For instances that require search space partitioning, the necessary depth scales substantially with $m$.
Verifying the most challenging inputs requires exploring fewer than 100 subproblems at $m = 1$, nearly 2000 at $m = 2$, and upwards of 4000 at $m = 4$.

\section{Related Work}\label{sec:related}

\paragraph{Bit-flip attacks and defenses.}Existing defenses against bit-flip attacks primarily rely on empirical techniques, such as triplicating critical weights~\citep{schmedding2024aspis}, planting honey neurons to attract and detect flips~\citep{liu2023neuropots}, or protecting weights with error-detecting codes~\citep{velcicky2026deepncode}. However, these approaches do not provide formal verification certificates. Furthermore, in our theoretical setting, introducing such error-detecting codes would merely alter the reachable set $R_n$ in \autoref{eq:reach} without fundamentally addressing the need for formal guarantees.

\paragraph{Verifying quantized networks.}Exact verifiers for quantized networks typically assume fixed weights and verify robustness against input perturbations. These methods rely on bit-precise SMT encodings~\citep{giacobbe2020how,henzinger2021quantized}, integer linear programming formulations like QVIP~\citep{zhang2022qvip} and EQV~\citep{huang2024towards}, or bounded model checking as in QNNVerifier~\citep{song2024automated}. Crucially, these frameworks target quantized activations, which differs from the BFAVerifier threat model adopted in our work. While BFAVerifier~\citep{zhang2025bfaverifier} introduces a symbolic-parameter domain and a MILP formulation to certify network weights, it is strictly limited to verifying flips in a single parameter. In contrast, our approach scales to verify simultaneous bit flips across any $m$ parameters.

\paragraph{Bound propagation and branch and bound.}State-of-the-art verification frameworks, such as CROWN~\citep{zhang2018crown,xu2020automatic,wang2021beta} and DeepPoly~\citep{singh2019deeppoly}, efficiently compute linear bounds via back-substitution. Verification beyond input-norm balls covers structural~\citep{duong2026verifying} and dual perturbations~\citep{duong2026verifying2}, whereas we perturb the parameters. The integrality property underpinning \autoref{prop:topm} is standard in classical combinatorial optimization~\citep{schrijver2003combinatorial}, and we leverage this property in every concretization step of our bound propagation. Finally, our second core contribution, branch and bound over perturbation locations, is orthogonal to existing activation-space branching methods.
Consequently, it could be naturally combined with strategies like the ReLU branching employed in $\beta$-CROWN~\citep{wang2021beta}, the DPLL(T)-style search with clause learning in NeuralSAT~\citep{duong2023dpll,duong2025neuralsat,duong2025neuralsat2}, or learned branching policies~\citep{duong2026verifying3}.
Generating and checking proofs of verification results~\citep{duong2026generating} certifies each verification run, whereas we mechanize the soundness of the method once in Lean.

\section{Conclusion}

We presented \tool{} to verify quantized networks against $m$-parameter bit-flip perturbations, surpassing the $m=1$ limit of BFAVerifier. By capturing $\binom{K}{m}$ choices via a single multi-bitflip bound propagation and branching over perturbation locations, \tool{} requires at most $1 + \binom{2\nu - L}{m}$ propagations, verifying 293 out of 625 problems.

The current limitations of \tool{} include the lack of support for certain operations such as max pooling and residual connections. Additionally, we observe a tradeoff between accuracy and robustness, with accuracy decreasing from 98.1\% to 96.2\% on the MNIST fully connected $2{\times}256$ architecture.
In future work, we plan to distribute the $m$ bit-flips across multiple layers, enabling verification of perturbations spread over several layers by solving parallel single-layer subproblems with fewer flips, in the spirit of compositional verification~\citep{duong2026compositional}. We also aim to incorporate tighter $\alpha$-CROWN~\citep{xu2020automatic} relaxations within each subproblem.

\bibliographystyle{unsrtnat}
\bibliography{paper}

\appendix

\section{Branch and Bound: Worked Example and Soundness}\label{app:bab}

This appendix works out a five-neuron example (\autoref{app:bab-example}), proves that the subproblems cover every perturbation (\autoref{app:bab-cover}), and proves \autoref{prop:bab-count} (\autoref{app:bab-count}).

\subsection{Worked Example: The Subproblems Cover the Integer Program}\label{app:bab-example}

\paragraph{Neural Network.}
Fix the input and take one hidden layer of five neurons with the margin
\begin{equation}\label{eq:ex-net}
  \hat z_j = 1 + d_j \quad (j = 1, \dots, 5), \qquad z_j = \mathrm{ReLU}(\hat z_j), \qquad
  f = z_1 + z_2 + z_3 + z_4 + z_5 - \tfrac52
\end{equation}
Only the biases $b_1, \dots, b_5$ are perturbable, so neuron $\hat z_j$ has the single parameter $b_j$ and nominal value $1$.
The change $d_j$ of $b_j$ lies in the change interval~\autoref{eq:hull}, here $[\ell, u] = [-3, 0.5]$, and at most $m = 2$ parameters flip.
With more parameters per neuron the argument is the same, and only the ranges of the $\hat z_j$ grow.

\paragraph{Integer program.}
A perturbation selects $s \in \mathcal{S}_2 = \{ s \in \{0,1\}^5 : \sum_j s_j \le 2 \}$ and changes $d_j \in [-3 s_j, 0.5 s_j]$, so an unselected bias keeps $d_j = 0$.
Problem~\eqref{eq:goal} becomes the integer program
\begin{equation}\label{eq:ex-ip}
  f^\star = \min_{s \in \mathcal{S}_2} F(s), \qquad
  F(s) = \min_{d_j \in [-3 s_j,\, 0.5 s_j]} \ \sum_j \mathrm{ReLU}(1 + d_j) - \tfrac52
       = 5 - \textstyle\sum_j s_j - \tfrac52
\end{equation}
since a perturbed neuron drops to $\mathrm{ReLU}(1 - 3) = 0$ and an unperturbed neuron stays at $1$.
The ten pairs are the worst of the 16 selections, so $f^\star = 0.5 > 0$ and the instance is verified.

\paragraph{Bounding.}
$\textsc{Bound}(S)$ is multi-bitflip bound propagation with only the biases in $U = \bigcup S$ perturbable.
For $j \notin U$ the range of $\hat z_j$ is $[1, 1]$, so $z_j = 1$ exactly.
For $j \in U$ the range $[-2, 1.5]$ makes the neuron unstable, and CROWN takes the lower line $z_j \ge 0$ since the upper line has slope $1.5/3.5 < 0.5$.
$\textsc{Bound}(S)$ is therefore the optimum of the linear program
\begin{equation}\label{eq:ex-lp}
  R(U) \;=\; \min\ \sum_j z_j - \tfrac52
  \quad \text{s.t.} \quad z_j \ge 0 \ (j \in U), \qquad z_j = 1 \ (j \notin U)
\end{equation}
Hence $R(U) = 5 - |U| - \tfrac52$, so a subproblem is proved exactly when it has at most two neurons.

At the root, where all five neurons are perturbable, $R(\{1,\dots,5\}) = -2.5 < 0$ although $f^\star = 0.5$.
In the integer program~\eqref{eq:ex-ip}, $z_j = 0$ needs $s_j = 1$, so at most two neurons drop, whereas the linear program~\eqref{eq:ex-lp} sets all five $z_j$ to $0$ at once.
This gap is the loss described in \autoref{sec:bab}, where the range of each neuron is computed as if it alone were perturbed.
The selection $s$ and the changes $d$ enter only through $\hat z_j$, which the line $z_j \ge 0$ does not involve, so the constraint $\sum_j s_j \le 2$ never reaches it.
Every term $g_j$ is therefore $0$, and the top-$2$ selection adds nothing.
\autoref{prop:topm} makes the selection exact given the terms, whereas the loss lies in the ReLU lines.

\paragraph{Branching.}
For a set of neurons $U$, let $\mathcal{S}_2(U)$ be the selections with $s_j = 0$ for $j \notin U$.
Every perturbation inside $U$ is feasible for~\eqref{eq:ex-lp}, since $z_j = \mathrm{ReLU}(\hat z_j) \ge 0$, so
\begin{equation}\label{eq:ex-sound}
  R(U) \;\le\; \min_{s \in \mathcal{S}_2(U)} F(s)
\end{equation}
The neurons $U_1, \dots, U_q$ of the subproblems \emph{cover} the integer program if $\mathcal{S}_2 = \mathcal{S}_2(U_1) \cup \dots \cup \mathcal{S}_2(U_q)$.
For a cover, taking the minimum of~\eqref{eq:ex-ip} over each $\mathcal{S}_2(U_t)$ and applying~\eqref{eq:ex-sound} to each gives
\begin{equation}\label{eq:ex-cover}
  f^\star \;=\; \min_{t \le q}\ \min_{s \in \mathcal{S}_2(U_t)} F(s) \;\ge\; \min_{t \le q} R(U_t)
\end{equation}
Proving $R(U_t) > 0$ for every $t$ thus proves $f^\star > 0$, even if the sets $\mathcal{S}_2(U_t)$ overlap.
Each subproblem is a smaller linear program with fewer relaxed ReLUs.
By \autoref{prop:cover-step}, the children that $\textsc{Branch}$~\eqref{eq:branch} returns cover their parent.

\paragraph{The search.}
The search proves the instance at depth 3 after 16 bound propagations (\autoref{fig:ex-tree5}).
Each subproblem shows its groups above its bound $R$~\eqref{eq:ex-lp}, in red if it fails, in green if it is proved, and dashed if it repeats an earlier one and is not bounded again.
Gray boxes show how $\textsc{Branch}$~\eqref{eq:branch} splits the largest group until $T$ has more than $m$ groups (1) and returns every set of $m$ groups of $T$ as a child (2).
The network has one perturbable layer, so the root is the single group $\{1,\dots,5\}$ and is branched directly.
Its first split gives $|T| = 2 = m$, so $\textsc{Branch}$ also splits $\{3,4,5\}$.
The neuron sets $\{1,2,3\}$, $\{1,2,4,5\}$, and $\{3,4,5\}$ of the three depth-1 subproblems contain all ten pairs, so they cover the integer program.
Every pair subproblem has $R = 0.5 > 0$ and is proved, and $\{\{4\},\{5\}\}$ is generated twice but bounded once.
All ten pairs are then proved and the answer is verified, with the bound $0.5 = f^\star$ by~\eqref{eq:ex-cover}.
The pair subproblems are exact here because they equal the worst selections.
In a deeper network, a pair subproblem still relaxes the ReLUs of its two neurons and of later layers, so its bound can stay below the integer optimum.

\begin{remark}[Lean mechanization of the worked example]\label{rem:lean-example}
The worked example is mechanized in \texttt{Proof/Example.lean}, with the five neurons indexed by \texttt{Fin 5}.
\texttt{F\_isLeast}, \texttt{fstar\_isLeast}, and \texttt{R\_eq} prove $F(s) = 5 - |s| - \tfrac52$ in~\eqref{eq:ex-ip}, $f^\star = 0.5$, and the optimum $R(U)$ of~\eqref{eq:ex-lp}.
\texttt{proved\_iff} proves that a subproblem is proved exactly when it has at most two neurons.
\texttt{ex\_sound} proves~\eqref{eq:ex-sound} assuming only $s \subseteq U$, and \texttt{ex\_cover} proves~\eqref{eq:ex-cover} for any finite nonempty family of subproblems that covers the selections.
\texttt{depth1\_cover} and \texttt{tab\_ex\_cover} check that the depth-1 subproblems cover the ten pairs, and \texttt{fstar\_cover} proves that the ten pair subproblems give the bound $0.5 = f^\star$.
The proofs mainly use \texttt{decide} on the finite set of selections, \texttt{rw}, and \texttt{norm\_num}, with \texttt{linarith} closing the inequalities.
They contain no \texttt{sorry} and introduce no axioms beyond Lean's standard \texttt{propext}, \texttt{Classical.choice}, and \texttt{Quot.sound}.
\end{remark}

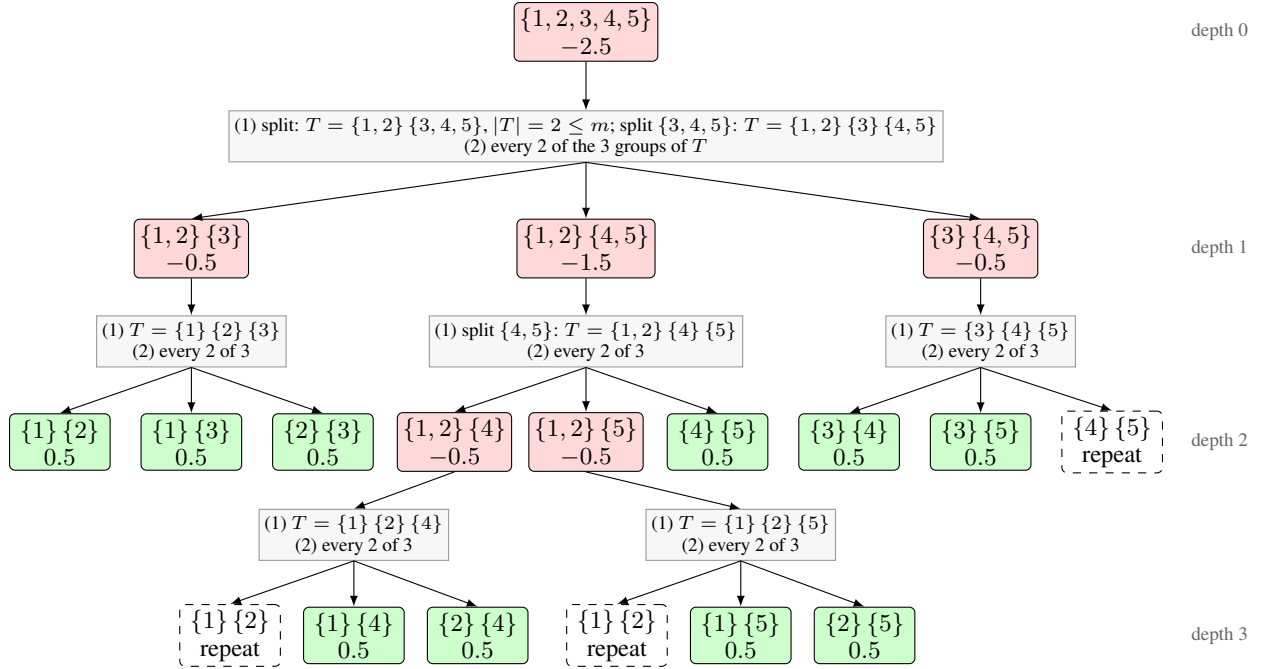
\begin{figure}[h]
\centering
\resizebox{\linewidth}{!}{%
\begin{tikzpicture}[font=\footnotesize,
    subproblem/.style={draw, rounded corners=2pt, inner sep=2pt, align=center, minimum width=13mm},
    fail/.style={subproblem, fill=red!15}, ok/.style={subproblem, fill=green!20},
    dup/.style={subproblem, dashed},
    split/.style={draw=black!40, fill=black!3, inner sep=2pt, align=center, font=\scriptsize},
    arr/.style={-latex}]
  \node[fail] (r) at (0, 0) {$\{1,2,3,4,5\}$\\$-2.5$};
  \node[split] (tr) at (0, -1.35)
    {(1) split: $T = \{1,2\}\,\{3,4,5\}$, $|T| = 2 \le m$;
     split $\{3,4,5\}$: $T = \{1,2\}\,\{3\}\,\{4,5\}$\\
     (2) every 2 of the 3 groups of $T$};
  \node[fail] (a) at (-5.1, -2.8) {$\{1,2\}\,\{3\}$\\$-0.5$};
  \node[fail] (b) at (0, -2.8)     {$\{1,2\}\,\{4,5\}$\\$-1.5$};
  \node[fail] (c) at (5.1, -2.8)  {$\{3\}\,\{4,5\}$\\$-0.5$};
  \node[split] (ta) at (-5.1, -4.0) {(1) $T = \{1\}\,\{2\}\,\{3\}$\\(2) every 2 of 3};
  \node[split] (tb) at (0, -4.0)     {(1) split $\{4,5\}$: $T = \{1,2\}\,\{4\}\,\{5\}$\\(2) every 2 of 3};
  \node[split] (tc) at (5.1, -4.0)  {(1) $T = \{3\}\,\{4\}\,\{5\}$\\(2) every 2 of 3};
  \node[ok]   (a1) at (-6.8, -5.3)  {$\{1\}\,\{2\}$\\$0.5$};
  \node[ok]   (a2) at (-5.1, -5.3) {$\{1\}\,\{3\}$\\$0.5$};
  \node[ok]   (a3) at (-3.4, -5.3)  {$\{2\}\,\{3\}$\\$0.5$};
  \node[fail] (b1) at (-1.7, -5.3) {$\{1,2\}\,\{4\}$\\$-0.5$};
  \node[fail] (b2) at (0, -5.3)     {$\{1,2\}\,\{5\}$\\$-0.5$};
  \node[ok]   (b3) at (1.7, -5.3)  {$\{4\}\,\{5\}$\\$0.5$};
  \node[ok]   (c1) at (3.4, -5.3)   {$\{3\}\,\{4\}$\\$0.5$};
  \node[ok]   (c2) at (5.1, -5.3)  {$\{3\}\,\{5\}$\\$0.5$};
  \node[dup]  (c3) at (6.8, -5.3)   {$\{4\}\,\{5\}$\\repeat};
  \node[split] (tb1) at (-3.0, -6.5) {(1) $T = \{1\}\,\{2\}\,\{4\}$\\(2) every 2 of 3};
  \node[split] (tb2) at (2.0, -6.5)  {(1) $T = \{1\}\,\{2\}\,\{5\}$\\(2) every 2 of 3};
  \node[dup]  (d1) at (-4.6, -7.8) {$\{1\}\,\{2\}$\\repeat};
  \node[ok]   (d2) at (-3.0, -7.8) {$\{1\}\,\{4\}$\\$0.5$};
  \node[ok]   (d3) at (-1.4, -7.8) {$\{2\}\,\{4\}$\\$0.5$};
  \node[dup]  (e1) at (0.4, -7.8)  {$\{1\}\,\{2\}$\\repeat};
  \node[ok]   (e2) at (2.0, -7.8)  {$\{1\}\,\{5\}$\\$0.5$};
  \node[ok]   (e3) at (3.6, -7.8)  {$\{2\}\,\{5\}$\\$0.5$};
  \draw[arr] (r) -- (tr);
  \foreach \c in {a, b, c} \draw[arr] (tr.south) -- (\c.north);
  \draw[arr] (a) -- (ta); \draw[arr] (b) -- (tb); \draw[arr] (c) -- (tc);
  \foreach \c in {a1, a2, a3} \draw[arr] (ta.south) -- (\c.north);
  \foreach \c in {b1, b2, b3} \draw[arr] (tb.south) -- (\c.north);
  \foreach \c in {c1, c2, c3} \draw[arr] (tc.south) -- (\c.north);
  \draw[arr] (b1.south) -- (tb1.north); \draw[arr] (b2.south) -- (tb2.north);
  \foreach \c in {d1, d2, d3} \draw[arr] (tb1.south) -- (\c.north);
  \foreach \c in {e1, e2, e3} \draw[arr] (tb2.south) -- (\c.north);
  \foreach \y/\d in {0/0, -2.8/1, -5.3/2, -7.8/3}
    \node[font=\scriptsize, text=black!60, anchor=west] at (7.7, \y) {depth \d};
\end{tikzpicture}}
\caption{Branch and bound on the five-neuron example.}
\label{fig:ex-tree5}
\end{figure}

\FloatBarrier

\subsection{Why the Cover Is Sound}\label{app:bab-cover}

A perturbation $\rho$ \emph{lies in} a subproblem $S$ when every parameter it changes belongs to $\bigcup S$.
Soundness of the verdict verified rests on one invariant, which the branching maintains.

\begin{proposition}[Children cover their parent]\label{prop:cover-step}
Let $S$ be a subproblem, $G \in S$ a group, and $G = G_1 \uplus G_2$ a split of $G$ into two disjoint non-empty parts, so that $T = (S \setminus \{G\}) \cup \{G_1, G_2\}$ has pairwise disjoint groups and $\bigcup T = \bigcup S$.
If $|T| > m$, every $(m,n)$-perturbation that lies in $S$ lies in one of the children $\{C \subseteq T : |C| = m\}$.
\end{proposition}
\begin{proof}
Let $\rho$ lie in $S$ and let $D \subseteq T$ be the groups that contain a parameter changed by $\rho$.
The groups of $T$ are pairwise disjoint, so each changed parameter lies in exactly one group of $T$.
Hence $|D| \le m$, because $\rho$ changes at most $m$ parameters.
Since $|T| \ge m$, choose any $C$ with $D \subseteq C \subseteq T$ and $|C| = m$.
This $C$ is one of the children, and every parameter changed by $\rho$ lies in $\bigcup D \subseteq \bigcup C$.
\end{proof}

\begin{remark}[Lean mechanization of \autoref{prop:cover-step} (Children cover their parent)]\label{rem:lean-cover-step}
\texttt{children\_cover} in \texttt{Proof/Cover.lean} proves \autoref{prop:cover-step} assuming only $m \le |T|$ and that the parameters changed by $\rho$ lie in $\bigcup T$.
It picks one group of $T$ for each changed parameter and extends these at most $m$ groups to $m$ groups of $T$.
\texttt{cover\_step} restates it for the split $T$, and \texttt{child\_ssubset} proves that a child of pairwise disjoint non-empty groups allows strictly fewer parameters than its parent.
The proofs mainly use \texttt{obtain}, \texttt{rw}, and induction over finite sets, contain no \texttt{sorry}, and add no axioms beyond \texttt{propext}, \texttt{Classical.choice}, and \texttt{Quot.sound}.
\end{remark}

If $|T| \le m$, the algorithm splits again before bounding anything.
Each split keeps the groups disjoint and $\bigcup T = \bigcup S$, and it adds one group to $T$.
After finitely many splits, either $|T| > m$ and \autoref{prop:cover-step} applies, or every group is a single neuron and the subproblem is a leaf.

\begin{proposition}[Cover invariant]\label{prop:cover}
At the start of every iteration of \autoref{alg:bab}, every $(m,n)$-perturbation lies in some subproblem of $Q \cup P$.
\end{proposition}
\begin{proof}
Let $\mathcal{G}$ be the set of subproblems generated so far, the initial subproblems and all children.
We prove a stronger invariant: every $(m,n)$-perturbation that lies in some $H \in \mathcal{G}$ lies in a subproblem of $P$, or in a subproblem $S \in Q$ with $\bigcup S \subseteq \bigcup H$.
\emph{Initially}, $\mathcal{G} = Q$, and $H$ itself is such a subproblem.
The layers are pairwise disjoint and cover all $K$ parameters, so the argument of \autoref{prop:cover-step} with $T = \{\text{layers}\}$ puts every perturbation in an initial subproblem.
If $L \le m$, the initial subproblems are those of $\textsc{Branch}$(all layers), which cover every perturbation by the same proposition.
The invariant therefore implies the proposition.
\emph{Step}: a subproblem $S$ leaves $Q$ in exactly three ways.
(i) $\textsc{Bound}(S) > 0$, and $S$ enters $P$.
(ii) The neurons of $S$ lie inside those of a proved $S' \in P$, so $\bigcup S \subseteq \bigcup S'$ and what lies in $S$ lies in $S'$.
(iii) $\textsc{Bound}(S) \le 0$, and $S$ is replaced by its children, which cover $S$ by \autoref{prop:cover-step}, and each child $C$ satisfies $\bigcup C \subsetneq \bigcup S$ because the groups are disjoint and non-empty.
A perturbation that lay in $S$ lies in some child $C$, which is queued if it is new.
If $C$ was generated before, the invariant for $C$ gives a proved subproblem, or a queued subproblem inside $\bigcup C \subsetneq \bigcup S$, which is therefore not $S$ and stays queued.
\end{proof}

\begin{remark}[Lean mechanization of \autoref{prop:cover} (Cover invariant)]\label{rem:lean-cover}
\texttt{cover\_invariant} in \texttt{Proof/Cover.lean} proves \autoref{prop:cover} assuming only that the initial subproblems cover every perturbation of at most $m$ parameters (\texttt{init\_cover}).
The search is modeled by the steps (i) to (iii), where a branch step queues only the children not generated before, and \texttt{branch\_step} proves that $\textsc{Branch}$ on pairwise disjoint non-empty groups is such a step.
\texttt{inv\_step} and \texttt{inv\_reach} prove that every step preserves the invariant of the proof above, which thus holds in every reachable state.
The proofs mainly use \texttt{obtain}, \texttt{by\_cases}, and induction over reachable states, contain no \texttt{sorry}, and add no axioms beyond \texttt{propext}, \texttt{Classical.choice}, and \texttt{Quot.sound}.
\end{remark}

\begin{proposition}[Soundness of the search]\label{prop:bab-sound}
If \autoref{alg:bab} returns verified, then $f^\star > 0$, and $\min_{S \in P} \textnormal{\textsc{Bound}}(S)$ is a lower bound on $f^\star$.
\end{proposition}
\begin{proof}
The algorithm returns verified only with $Q = \emptyset$, so by \autoref{prop:cover} every $(m,n)$-perturbation $\rho$ lies in some $S \in P$.
$\textsc{Bound}(S)$ is multi-bitflip bound propagation with only $\bigcup S$ perturbable with at most $m$ flipped parameters.
By \autoref{sec:topm}, $\textsc{Bound}(S)$ is a lower bound on $f^{\rho'}(x)$ for every $x \in \mathcal{B}$ and every $(m,n)$-perturbation $\rho'$ lying in $S$, and in particular for $\rho$.
Since $\textsc{Bound}(S) > 0$ for every $S \in P$, the claim follows.
\end{proof}

\begin{remark}[Lean mechanization of \autoref{prop:bab-sound} (Soundness of the search)]\label{rem:lean-bab-sound}
\texttt{search\_sound} and \texttt{search\_lower\_bound} in \texttt{Proof/Cover.lean} prove \autoref{prop:bab-sound} from \texttt{cover\_invariant}, assuming covering initial subproblems, a stop with $Q = \emptyset$, and that $\textsc{Bound}(S)$ bounds the margin from below under every perturbation that lies in $S$, which is the soundness of multi-bitflip bound propagation (\autoref{rem:lean-terms}).
\texttt{search\_sound} adds $\textsc{Bound}(S) > 0$ for $S \in P$ and concludes a positive margin, and \texttt{search\_lower\_bound} adds $P \neq \emptyset$ and concludes the lower bound $\min_{S \in P} \textsc{Bound}(S)$.
The proofs mainly use \texttt{obtain}, contain no \texttt{sorry}, and add no axioms beyond \texttt{propext}, \texttt{Classical.choice}, and \texttt{Quot.sound}.
\end{remark}

Nothing is claimed for the other two verdicts: unknown reports the bound of the failing leaf, and timeout reports the bound of the root, which holds with all $K$ parameters perturbable.

A ReLU phase split~\citep{wang2021beta} \emph{partitions} its parent, whereas our children \emph{cover} it, since a perturbation of fewer than $m$ parameters lies in several of them.
A leaf of fixed phases is exact, but a leaf of single neurons is not, because its ReLU lines can still be loose.

\paragraph{Refinement.}
A child allows a subset of its parent's parameters, yet its bound is occasionally worse, since CROWN's lower ReLU line (the slope $0$ or $1$ nearest the upper line) can flip when a recomputed neuron range tightens.
A repository script compares 5{,}124 parent--child pairs from eight networks (three refinement levels, $m \in \{1,2,3\}$, at a point and on a box), and in every pair the child's parameter set is a subset of its parent's.
In 63 pairs (1.2\%), the child's bound is worse, and the largest drop is $4.9$.
This costs bound propagations, never soundness, because only failing subproblems are split, no parent bound is reused, and \autoref{prop:bab-sound} takes the minimum over independently proved subproblems.

A repository test drives the search to the leaves with a stub that proves exactly the subproblems of at most $m$ neurons (or channels), and checks on dense and convolutional networks that \emph{every} set of at most $m$ neurons lies in a proved subproblem.
Another test compares $\textsc{Bound}(S)$ at a point input with the exact minimum over all perturbations of at most $m$ parameters in $\bigcup S$, and over 32 (subproblem, network) pairs no bound exceeds it.
Of these bounds, 4 are tight to machine precision, so on those subproblems the multi-parameter concretization is exact.
\autoref{sec:eval} adds the end-to-end check that no verified answer contradicts a counterexample of a baseline.

\subsection{Proof of \autoref{prop:bab-count}}\label{app:bab-count}

\begin{proof}
\emph{Upper bound.} The \emph{splitting tree} of layer $i$ has the whole layer as its root, and every group of at least two neurons has its two parts as children.
Its leaves are the $\nu_i$ single neurons and every inner node has two children, so it has $2\nu_i - 1$ nodes, and the $L$ trees together have $2\nu - L$ nodes.
Every group of the search is such a node, since the first groups are whole layers and $\textsc{Branch}$ replaces a node by its children.
Every subproblem bounded after the root is a set of exactly $m$ nodes, since the first subproblems are sets of $m$ layers (if $L \le m$, those of $\textsc{Branch}$(all layers)) and $\textsc{Branch}$ returns only subsets of $m$ groups of $T$.
A subproblem already generated is not queued again, so each set of $m$ nodes is bounded at most once.
Therefore, at most $\binom{2\nu - L}{m}$ subproblems are bounded after the root, and the root adds one.
\emph{Lower bound.} Since $\nu > m$, the root fails.
A subproblem bounded after the root has $m$ non-empty groups, so it is proved exactly when every group is a single neuron, and a failing subproblem has a group of at least two neurons.
One split gives $|T| = m + 1$, so $\textsc{Branch}$ returns children.
If $L \le m$, the first $\textsc{Branch}$ splits the layers, and as long as $|T| \le m < \nu$ some group has at least two neurons, so it also reaches $|T| = m + 1$.
The search thus never returns unknown.
There is no time limit, and by the upper bound only finitely many subproblems are queued, each once, so $Q$ empties and the search returns verified.
By \autoref{prop:cover}, the proved subproblems then cover every perturbation.
A perturbation that changes one parameter in each of $m$ given neurons lies only in the proved subproblem made of exactly these neurons, so all $\binom{\nu}{m}$ such subproblems are bounded.
For $m = 1$, a failing subproblem is one node, and its children are the node's two parts.
No node is skipped, since the proved subproblems are single neurons, so every node is bounded once.
With $L \ge 2$, the first subproblems are the $L$ layers, so the search makes exactly $1 + 2\nu - L$ bound propagations.
\end{proof}

For a fixed $m$, both bounds grow as $\nu^m$, and for $\nu \gg m$ they differ by a factor of about $(2 - L/\nu)^m$, close to $2^m$ when the layers are wide.
A repository test runs the search with the stub of \autoref{app:bab-cover} on five networks for $m \le 3$ and checks both bounds and the equality for $m = 1$.

\begin{remark}[Lean mechanization of \autoref{prop:bab-count} (Number of bound propagations)]\label{rem:lean-bab-count}
The upper bound is mechanized in \texttt{Proof/Subproblems.lean}, assuming only $\nu_i \ge 1$.
\texttt{card\_splitTree}, \texttt{splits\_mem}, and \texttt{card\_groupUniv\_le} prove that the splitting tree of a layer, built as in the implementation, has $2\nu_i - 1$ nodes and is closed under splitting, and that the $L$ trees have at most $2\nu - L$ nodes.
\texttt{subproblems\_card\_le} then bounds the number of distinct subproblems generated by $\textsc{Branch}$ by $\binom{2\nu - L}{m}$.
The counting step of the lower bound is \texttt{stub\_proves\_m\_sets} and \texttt{stub\_card\_ge} in \texttt{Proof/Cover.lean}.
When the search stops with $Q = \emptyset$ and every proved subproblem has at most $m$ neurons, every set of $m$ neurons is the neuron set of a proved subproblem, so at least $\binom{\nu}{m}$ subproblems are proved.
The implementation supplies the remaining facts: no subproblem is queued twice, and the search terminates.
The proofs mainly use \texttt{omega}, \texttt{rw}, \texttt{obtain}, and functional induction on the splitting tree, contain no \texttt{sorry}, and add no axioms beyond \texttt{propext}, \texttt{Classical.choice}, and \texttt{Quot.sound}.
\end{remark}

\FloatBarrier

\section{Proof of \autoref{prop:topm} and Lean Mechanizations}\label{app:proof}

\begin{proof}
\emph{Step 1.} Pick $\theta \le 0$ with $g_p \le \theta$ for $p \in T$ and $g_p \ge \theta$ for $p \notin T$.
If some parameter lies outside $T$, then $\theta = \min\bigl(0, \min_{q \notin T} g_q\bigr)$ works, since $g_p \le \min(0, g_q)$ for $p \in T$ and $q \notin T$.
Otherwise $\theta = 0$ works.
\emph{Step 2.} For every $s \in \mathcal{C}_m$ and every $p$,
\begin{equation}
  (s_p - \mathbf{1}_T(p))\, g_p \;\ge\; (s_p - \mathbf{1}_T(p))\, \theta
\end{equation}
because $s_p - 1 \le 0$ and $g_p \le \theta$ for $p \in T$, and $s_p \ge 0$ and $g_p \ge \theta$ for $p \notin T$.
Summing over $p$ gives
\begin{equation}
  \sum_p s_p g_p - \sum_{p \in T} g_p \;\ge\; \theta \Bigl(\sum_p s_p - |T|\Bigr)
\end{equation}
\emph{Step 3.} The right-hand side is $\ge 0$, which proves (a).
If some parameter lies outside $T$, then $|T| < K$, so $|T| = \min(m,K) = m$, and both $\theta$ and $\sum_p s_p - |T|$ are $\le 0$.
Otherwise $\theta = 0$.
\emph{Step 4.} The 0/1 vector $\mathbf{1}_T$ has $\sum_p \mathbf{1}_T(p) = |T| \le m$, so it is in $\mathcal{S}_m \subseteq \mathcal{C}_m$ with objective $\sum_p \mathbf{1}_T(p)\, g_p = \sum_{p \in T} g_p$, which proves (b).
By (a), (b), and $\mathcal{S}_m \subseteq \mathcal{C}_m$, the three minima are equal.
\end{proof}

\begin{remark}[Lean mechanization of \autoref{prop:topm}]\label{rem:lean-topm}
\autoref{prop:topm} is mechanized in Lean~4 with Mathlib as \texttt{topm\_le\_lp} (part (a)) and \texttt{lp\_relaxation\_exact} (the three equal minima) in \texttt{Proof/TopM.lean}, whose only hypothesis on the terms is $g_p \le 0$.
They characterize $T$ by $|T| = \min(m,K)$ and $g_p \le g_q$ for $p \in T$ and $q \notin T$, so no sorting is needed, and \texttt{exists\_topm} constructs such a $T$ for every $g$.
The threshold $\theta$ is a dual certificate (\texttt{le\_of\_certificate}) reducing the proof to elementwise linear inequalities.
\texttt{topm\_sound} proves the soundness step of \autoref{sec:budget}, that a perturbation of at most $m$ parameters whose terms are at least $g_p$ changes the bound by at least $\topm(g)$.
The proofs mainly use \texttt{rw}, \texttt{refine}, and \texttt{linarith}, contain no \texttt{sorry}, and add no axioms beyond \texttt{propext}, \texttt{Classical.choice}, and \texttt{Quot.sound}.
\end{remark}

\begin{remark}[Lean mechanization of the soundness of multi-bitflip bound propagation]\label{rem:lean-terms}
In \texttt{Proof/Terms.lean}, a corner attains the minimum in~\eqref{eq:g} over any nonempty rectangle (\texttt{corner\_isLeast}), $\underline g_p \le 0$ follows from $\ell_p \le 0 \le u_p$ alone (\texttt{corner\_nonpos}), and minimizing a term linear in the change over the interval hull of $R_n$ gives its minimum over $R_n$ (\texttt{hull\_min\_isLeast}).
\texttt{sep\_isLeast}, \texttt{concretize\_sound}, and \texttt{concretize\_exact} prove with \texttt{topm\_sound} that $\min_{x \in \mathcal{B}} (Ax + c) + \topm(g)$ is a lower bound on the margin under every flip in at most $m$ parameters and the minimum of the relaxed program.
For a convolution, \texttt{sum\_min\_attained} proves that the range of $S_p$ is the sum of the ranges of its summands, and \texttt{conv\_corner\_le} applies the corner bound to $D_p S_p$.
The proofs mainly use \texttt{rcases}, \texttt{refine}, and \texttt{nlinarith}, contain no \texttt{sorry}, and add no axioms beyond \texttt{propext}, \texttt{Classical.choice}, and \texttt{Quot.sound}.
\end{remark}

\begin{remark}[Lean mechanization of the quantization facts]\label{rem:lean-quant}
In \texttt{Proof/Quant.lean}, \texttt{card\_reach\_real} proves that $R_n$ has exactly $\sum_{r \le n} \binom{Q}{r}$ elements, assuming only $\Delta_i \ne 0$, and \texttt{change\_interval\_contains\_zero} proves $\ell_\alpha \le 0 \le u_\alpha$ in~\eqref{eq:hull}.
\texttt{sign\_flip\_min\_real} proves that one sign-bit flip reaches $-2^{Q-1}\Delta_i$, and \texttt{clip\_change\_bound} proves the bound $|D^{(i)}_{jk}| \le (2^Q-1)\,c/(2^{Q-1}-1)$ used in fine-tuning (\autoref{app:setup}), assuming only that every real parameter of the layer lies in $[-c, c]$.
The proofs mainly use \texttt{rw}, \texttt{simp}, and \texttt{refine}, contain no \texttt{sorry}, and add no axioms beyond \texttt{propext}, \texttt{Classical.choice}, and \texttt{Quot.sound}.
\end{remark}
\section{Experimental Details}\label{app:setup}

\paragraph{Networks.}
The architectures (\autoref{tab:nets}) and ONNX files come from the VNN-COMP benchmarks~\citep{brix2023first,kaulen20256th} \texttt{acasxu}, \texttt{mnistfc}, and \texttt{cora}.
Our loader matches \texttt{onnxruntime} on 64 random inputs to a relative error below $10^{-6}$, the float32 precision.
The \texttt{cora} networks read images with swapped spatial axes, MNIST as $(W, H)$ and CIFAR as $(C, W, H)$.
With this layout, their accuracy is 98.2\% and 51.3\%, respectively, against 13.6\% and 20.6\% with the usual layout.
For ACAS Xu, accuracy is the agreement with the original network's advisory on sampled inputs.
After plain 8-bit quantization, the published networks agree with only 7.9\% to 58.2\% of these advisories (pretrained column), so BFAVerifier also retrained them.

\begin{table}[h]
\centering
\small
\begin{tabular}{llrrcc}
\toprule
network & layers & $K$ & clip $c$ & $Q$-bit accuracy & pretrained \\
\midrule
ACAS Xu 1\_1 & $6{\times}50$ & 13{,}305 & 0.5 & 94.4\% & 7.9\% \\
ACAS Xu 2\_1 & $6{\times}50$ & 13{,}305 & 0.5 & 90.0\% & 58.2\% \\
ACAS Xu 3\_3 & $6{\times}50$ & 13{,}305 & 0.5 & 94.3\% & 17.9\% \\
MNIST FC $2{\times}256$ & $2{\times}256$ & 269{,}322 & 0.1 & 96.2\% & 98.1\% \\
MNIST FC $4{\times}256$ & $4{\times}256$ & 400{,}906 & 0.1 & 97.9\% & 97.7\% \\
MNIST FC $7{\times}250$ & $7{\times}250$ & 575{,}260 & 0.05 & 97.6\% & 98.2\% \\
CIFAR FC $7{\times}250$ & $7{\times}250$ & 1{,}147{,}260 & 0.05 & 51.9\% & 51.3\% \\
\bottomrule
\end{tabular}
\caption{The evaluated networks.}
\label{tab:nets}
\end{table}

\paragraph{Fine-tuning.}
Fine-tuning starts from the pretrained weights and minimizes the hinge loss $\max(0, M - \text{margin})$ with $M = 5$, batch size 128, and 10 epochs of Adam, clipping every parameter to $[-c,c]$ after each step.
ACAS Xu has no training data, so its labels are the original network's advisories on 200k inputs sampled uniformly from its normalized domain.
Before any verification run, $c$ is fixed as the smallest value whose $Q$-bit accuracy is within 3 points of the pretrained network.
The grids are $c \in \{0.05, 0.1, 0.2\}$ for the image networks and $\{0.5, 1, 2\}$ for ACAS Xu, each with learning rates $10^{-3}$ and $10^{-4}$, and are extended upward when no value qualifies.
For ACAS Xu no such $c$ survives a single flip, so we take the most accurate setting whose three pilot instances at $m = 1$ are not all falsified (a counterexample checked by an exact forward pass).
Every intermediate result is in the repository's logs.

Clipping bounds every change interval uniformly.
With every real parameter in $[-c, c]$, the step is $\Delta_i \le c/(2^{Q-1}-1)$, and flipping any bits of a $Q$-bit code changes the integer by at most $2^Q-1$.
Hence $|D^{(i)}_{jk}| \le (2^Q-1)\,c/(2^{Q-1}-1) \approx 2c$ for every parameter, whatever its value and whichever bits flip, including the sign bit.

\paragraph{Inputs.}
For the image networks, we permute the test set with seed 0, take the first 15 correctly classified images that bound propagation proves at $m = 0$, and clip their $\ell_\infty$ boxes of radius $0$ or $2/255$ to $[0,1]$.
At radius $2/255$, no CIFAR-10 test image passes this filter (4{,}809 are misclassified and 5{,}191 are not proved at $m = 0$), so CIFAR-10 has instances only at radius $0$.
ACAS Xu inputs are sampled uniformly from the normalized domain, with radius $0$ or $0.01$ and the same $m = 0$ filter.

\paragraph{The standard (pretrained) networks.}
With the published weights, no method proves any instance of MNIST FC $2{\times}256$ or ACAS Xu 1\_1 (radius $0$, $m \in \{1, 2\}$, 10 inputs, 40 instances in total).
\tool{} returns unknown, and both baselines time out or find a counterexample, which motivates fine-tuning (\autoref{sec:eval}).

\paragraph{Hardware and software.}
All runs use one machine with an Intel Core i9-9900K CPU (8 cores), 62\,GB of memory, and an NVIDIA RTX~2080~Ti GPU (11\,GB).
Multi-parameter bound propagation and BFAVerifier's oracle run on the GPU in 64-bit floating point, and the MILP on the CPU with Gurobi~13 and 8 threads.
Every instance has a limit of 60\,s, and timed jobs run one at a time.

\paragraph{Implementation.}
We keep only a running top-$m$ of the terms of each neuron and bound all neurons of a layer together as one matrix, with convolutions as in \autoref{app:conv}.
Branch and bound bounds a batch of subproblems in one GPU call, each as a 0/1 mask over the $K$ parameters.
In a dense layer, the terms $\underline g_p$ of the weights of one neuron $\hat z_j$ differ only by the factor $\underline A_j$, so only the $m$ most extreme terms of each neuron can enter the top-$m$.
Keeping them instead of the full $(S, \text{out}, \text{in})$ tensor makes the $7{\times}250$ networks practical and leaves every bound bit-identical, as a test checks.

\FloatBarrier

\subsection{The Exact MILP Baseline}\label{app:milp}

The baseline encodes the verification problem~\eqref{eq:goal} exactly for at most $m$ flipped parameters.
Each perturbable parameter $p$ has binaries $\delta_{p,r}$ with $\sum_r \delta_{p,r} = 1$, one per reachable value $v_{p,r} \in R_n(p)$ of~\eqref{eq:reach}, with $v_{p,0}$ nominal.
The constraint $\sum_p (1 - \delta_{p,0}) \le m$ limits the number of flipped parameters.
The product of a perturbed weight $W^{(i)}_{jk}$ and its input $z_k \in [\underline z_k, \overline z_k]$ is linearized with continuous $y_{r}$, $\underline z_k\, \delta_{p,r} \le y_{r} \le \overline z_k\, \delta_{p,r}$ and $\sum_r y_{r} = z_k$.
It contributes $\sum_r v_{p,r}\, y_{r}$ to $\hat z_j$, which is exact because exactly one $\delta_{p,r}$ is $1$.
A perturbed bias adds $\sum_r (v_{p,r} - v_{p,0})\, \delta_{p,r}$ to its nominal value.
The copies of a kernel entry have separate products but share its binaries $\delta_{p,\cdot}$, so the constraint counts the entry once.
Average pooling and other fixed linear maps enter as linear equalities, and the input ranges over the box $\mathcal{B}$.
A ReLU with input range $[l, u]$ is encoded exactly as $z = \hat z$ if $l \ge 0$ and as $z = 0$ if $u \le 0$.
Otherwise, a binary $\beta$ enforces $z \ge \hat z$, $z \ge 0$, $z \le u \beta$, and $z \le \hat z - l(1-\beta)$.
The bounds $[l,u]$ come from multi-bitflip bound propagation for the same $m$, so they hold for every perturbation.

For each margin row $c$ of $C$, Gurobi minimizes $c^\top N^\rho(x)$ and stops once its best bound exceeds $10^{-6}$, which proves the row, or once it finds a solution below $-10^{-6}$.
The instance is verified only when the solver proves every row, and the rows share the time limit, including model building.
With one set of binaries per parameter and one product per parameter occurrence, the model does not build within the limit for hundreds of thousands of parameters.
If the first $m = 1$ instance of a run ends during model building, the remaining MILP instances of the run are skipped and reported as unproved.

Our re-implementation of BFAVerifier~\citep{zhang2025bfaverifier} follows its Algorithm~1.
For every parameter, the procedure BFA\_RA splits the sorted reachable values into non-negative and negative parts and bisects each with a bound-propagation oracle in which the parameter is symbolic.
The parameters that BFA\_RA cannot prove go to the exact MILP with $m = 1$, and the GPU oracle answers the first query of every parameter in batches.
For a convolution, as in Section~4.5.2 of that paper, the kernel entry is symbolic in all its copies, and each copy gets a separate linear bound.
On 18 small dense networks and 36 small convolutional instances, both baselines agree with the ground truth from enumerating every perturbation.

\begin{remark}[Lean mechanization of the exact MILP encoding]\label{rem:lean-milp}
In \texttt{Proof/Milp.lean}, \texttt{relu\_exact} proves that the big-M constraints hold for some binary $\beta$ if and only if $z = \max(\hat z, 0)$, assuming only $l \le \hat z \le u$.
For binaries in $\{0,1\}$ that satisfy the product constraints, \texttt{onehot\_product} proves that the product contributes exactly $v_{p,r} z_k$ for the selected value, since every other $y_r$ is $0$.
\texttt{onehot\_feasible} proves that every value is feasible when $\underline z_k \le z_k \le \overline z_k$.
\texttt{budget\_iff} proves that $\sum_p (1 - \delta_{p,0}) \le m$ holds if and only if at most $m$ parameters are perturbed, assuming only $\delta_{p,0} \in \{0,1\}$.
The proofs mainly use \texttt{rw}, \texttt{simp}, and \texttt{linarith}, contain no \texttt{sorry}, and add no axioms beyond \texttt{propext}, \texttt{Classical.choice}, and \texttt{Quot.sound}.
\end{remark}

\subsection{Convolutions and Pooling}\label{app:conv}

A convolution is stored as its kernel with a map from every nonzero entry of the unrolled matrix to the kernel entry it copies, and back-substitution uses a transposed convolution.
The term of kernel entry $p$ sums the interval of $A_{j_e} z_{k_e}$ over its copies and takes the four corners of \autoref{sec:topm}.
An unfold followed by a matrix product computes these sums for all bounds at once.
In the first layer, neuron $(c, q)$ at position $q$ of channel $c$ depends only on the kernel entries and bias of channel $c$.
Its range directly adds the $m$ most extreme of their terms, which gives the same range as back-substitution with $A = I$ without forming the identity.
Average pooling and other fixed linear maps between a ReLU and the next layer have no parameters, so they map ranges by their interval image and enter back-substitution as sparse matrices.

Sharing is essential for soundness.
Take a network with a $1{\times}4{\times}4$ input, a $2{\times}2$ convolution with two channels, average pooling, and a dense layer ($Q = 8$, $n = 1$, one input point).
At $m = 1$, treating the copies of each kernel entry as separate parameters gives a lower bound of $+4.94$ on a margin whose true worst value under one kernel flip is $+4.15$.
Our bound is $+3.63$.
A single flipped kernel entry changes all of its copies at once, and the separate-copies model counts this change as several perturbed parameters.

\end{document}